\documentclass[letterpaper,11pt]{article}
\usepackage{fullpage}

\usepackage[style=alphabetic,maxnames=10,backref=true]{biblatex}
\newcommand{\citet}[1]{\citeauthor*{#1}~\cite{#1}}

\usepackage[utf8]{inputenc} 
\usepackage[T1]{fontenc}    
\usepackage[hidelinks,colorlinks,allcolors=blue]{hyperref}        
\usepackage{url}            
\usepackage{booktabs}       
\usepackage{amsfonts}       
\usepackage{nicefrac}       
\usepackage{microtype}      
\usepackage[dvipsnames]{xcolor}        
\usepackage{graphicx}

\usepackage[linesnumbered,ruled,vlined]{algorithm2e}

\newcount\Comments  
\newcommand{\kibitz}[2]{\ifnum\Comments=1{\color{#1}{#2}}\fi}

\title{Sequential Blocked Matching}

\author{%
   Nicholas Bishop \\
  University of Southampton, UK\\
  \texttt{nb8g13@soton.ac.uk} 
  \and
  Hau Chan \\
  University of Nebraska-Lincoln, USA\\
  \texttt{hchan3@unl.edu} 
  \and
  Debmalya Mandal \\
  Max Planck Institute for \\Software Systems, Germany\\
  \texttt{dmandal@mpi-sws.org} 
  \and
  Long Tran-Thanh \\
  University of Warwick, UK\\
  \texttt{long.tran-thanh@warwick.ac.uk} \\
}
\usepackage{amsthm}
\usepackage{amsmath}
\usepackage{newfloat}
\usepackage{listings}
\usepackage{pifont}

\newtheorem{defn}{Definition}
\newtheorem{thm}{Theorem}
\newtheorem{lemma}{Lemma}

\newcommand{\argmax}{\text{argmax}}
\newcommand{\E}{\mathbb{E}}
\newcommand{\set}[1]{\left\{ #1 \right\}}
\newcommand{\opt}{\mathrm{OPT}}
\newcommand{\baserank}{$\mathrm{BASE}$}
\newcommand{\SW}{\textrm{SW}}

\begin{document}
\SetKwInput{KwInput}{Input}                
\SetKwInput{KwOutput}{Output}
\maketitle

\begin{abstract}
  We consider a sequential blocked matching (SBM) model 
  where strategic agents repeatedly report ordinal preferences over a set of services to a central planner. 
  The planner's goal is to elicit agents' true preferences and design a policy that matches services to agents in order to maximize the expected social welfare with the added constraint that 
  each matched service can be \emph{blocked} or unavailable for a number of time periods. 
  Naturally, SBM models the repeated allocation of reusable services to a set of agents where each allocated service becomes unavailable for a fixed duration. 

  We first consider the offline SBM setting, where the strategic agents are aware of their true preferences.  
  We measure the performance of any policy by \emph{distortion}, the worst-case multiplicative approximation guaranteed by any policy.
  For the setting with $s$ services, we establish lower bounds of $\Omega(s)$ and $\Omega(\sqrt{s})$ on the distortions of any deterministic and randomised mechanisms, respectively. 
  We complement these results by providing approximately truthful, measured by \emph{incentive ratio}, deterministic and randomised policies based on random serial dictatorship which match our lower bounds. Our results show that there is a significant improvement if one considers the class of randomised policies. 
  Finally, we consider the online SBM setting with bandit feedback where each agent is initially unaware of her true preferences, and the planner must facilitate each agent in the learning of their preferences through the matching of services over time.
  We design an approximately truthful mechanism based on the Explore-then-Commit paradigm, which achieves logarithmic dynamic approximate regret.
\end{abstract}


\section{Introduction}

In recent years, machine learning algorithms have been extremely successful in various domains, from playing games to screening cancer. However, despite such success, most learning algorithms cannot be deployed directly in practice to make decisions under uncertainty. The main reason is that most real-world applications involve multiple agents, and learning algorithms are often constrained due to the unavailability of resources. Motivated by such constraints in multi-agent systems, we consider the problem of repeated matching with blocking constraints, a scenario where multiple agents simultaneously learn their preferences with repeated blocking or unavailability of resources.

In particular, we are interested in the repeated \emph{one-sided matching} problem where strategic agents report their ordinal preferences based on their expected rewards over a set of alternatives, or \emph{services}. The agents are matched to the services given their reported preferences each time period or round. It is well-known that one-sided matching can be used to model various real-world situations such as matching patients to kidneys across health institutions  \cite{Durlauf:2008uq,Roth:2004ue}, assigning students to rooms in residential halls \cite{Durlauf:2008uq}, allocating workers to tasks in crowdsourcing  \cite{Difallah:2013uo,Aldhahri:2015vb}, and recommending users to activities in recommender systems  \cite{Satzger:2006tp,Ansari:2000uv,Isinkaye:2015vv}. 

In many of these situations, there are several obstacles. First, for a setting with reusable services, 
a major caveat is that an agent-alternative match within a round can result in 
\emph{blocking} of some services in which 
the services may not be available until a later time. 
For example, a recommended activity (e.g., a special promotion offer from a restaurant) that is matched to (or used by) a user may not be available to all users again until a later time. Or, in cloud computing, where tasks are matched to resources (e.g. GPUs), once a task is assigned to a resource, that resource is blocked for a certain number of rounds.  


Second, the agents are often unaware of their exact preferences, and the planner must coordinate their \emph{explorations} without incurring a significant loss. This is often true for recommending restaurants to customers, as the restaurants have limited capacity and people are rarely informed of all possible choices~\cite{Wald08}. Note that, even when the agents are themselves using learning algorithms over time, coordination by the planner becomes necessary to avoid different agents exploring the same service at a time -- a problem which is exacerbated by blocking of the services.

Finally, in several settings, the agents are aware of their preferences, but they might be \emph{strategic} in reporting their preferences for getting matched to better services. This is particularly prominent in assigning rooms to students. Rooms can be blocked due to occupancy and can be made available again once the students leave the rooms. As a result, there is a potential for the students to misreport their private preferences to manipulate matching outcomes.

\begin{table*}[!t]
\label{tab:1}
\setlength{\tabcolsep}{3pt}
\center
\caption{Lower and Upper Bound Results for Offline SBM Models\label{table:results_OSBM}}
\begin{tabular}{c|cc} \hline
 & Distortion & Incentive Ratio  \\ \hline
Any Deterministic Mechanism (\emph{lower bound}) & $\Omega(s)$ & $(0,1]$  \\  
Derandomized Repeated Random Serial Dictatorship (\emph{upper bound}) & $O(s)$ & $(1-1/e)$ \\ \hline
Any Randomised Mechanisms (\emph{lower bound}) & $\Omega(\sqrt{s})$ & $(0,1]$  \\  
Repeated Random Serial Dictatorship  (\emph{upper bound}) & $O(\sqrt{s})$ & $(1 - 1/e)$  \\ \hline
\end{tabular}
\end{table*}%

\subsection{Main Contributions}
In order to capture the notion of one-sided matching with blocking, we introduce a \emph{sequential blocked matching} (SBM) model, in which a set of $n$ strategic agents are matched to a set of $s$ services repeatedly over rounds 
 and where matched services are \emph{blocked} for a deterministic number of time steps. Each agent reports her ordinal preferences over the services every round, based on her current estimates of the expected rewards of the services. As is standard in the matching literature, we focus on the setting where agents just report ordinal preferences over the services. 
The planner's goal is to derive a matching policy that, at each round, 
elicits true preferences from the agents and matches them to services in order to maximize the expected social welfare, 
which is the sum of the expected utilities of the agents from the matchings over time, whilst accounting for the blocking of services.
To the best of our knowledge, SBM models have not been studied before and can be applied to a wide range of real-world matching scenarios. 


We 
investigate the offline and online variations of the SBM model. 
For both variations, we are interested in deriving  deterministic and randomized policies that are approximately truthful and efficient. 
We measure truthfulness by \emph{incentive ratio} \cite{IR}, which measures how much a single agent can gain via misreporting preferences. We measure efficiency through the notion of \emph{distortion} from social choice theory \cite{PR06}, which measures the loss in social welfare due to access to only preferences, and not utility functions and rewards. We formally define these concepts in Section \ref{sec:offline-blockmatch}.

\textbf{Offline SBM Benchmarks.} In the offline setting of SBM, each agent knows their own preferences and rewards over the services, but the planner does not. In addition, each agent reports their preferences  only once to the planner, before matching begins. Essentially, the offline benchmarks establish what we can achieve in terms of distortion if the agents' don't have to learn.  Table~\ref{table:results_OSBM} summarizes our results.
More specifically, we derive  lower bounds on the distortion of any deterministic and randomised mechanism. The main ingredient of our proof is the careful construction of reward profiles that are consistent with reported preferences that guarantees poor social welfare for the planner. 
We then focus on the upper bound and provide approximately truthful mechanisms with bounded incentive ratios that match the distortion lower bounds. In short, both the deterministic and randomised mechanisms we provide are based on the repeated random serial dictatorship (RSD) mechanism for one-shot one-sided matching problems. Our randomised mechanism, repeated RSD (\texttt{RRSD}), iterates randomly over all agents, greedily choosing the current agents' preferred service at each time step. Our deterministic mechanism, derandomised \texttt{RRSD} (\texttt{DRRSD}), is a derandomised version of this algorithm and matches the corresponding lower bound. Interestingly, we find that there is a strict separation of $\sqrt{s}$ between the achievable distortion by a deterministic and randomized mechanism. 


\textbf{Online SBM Algorithms.} 
For the online setting of SBM, the agents do not know their  preferences or rewards  
and must learn their preferences via repeated matching to services. After each matching, the agents update their preferences and strategically report them to the planner.  
We design an approximately truthful mechanism, bandit \texttt{RRSD} (\texttt{BRRSD}), based on the Explore-then-Commit (ETC) paradigm, which achieves sublinear dynamic approximate regret.
In particular, \texttt{BRRSD} has two phases. In the first phase, it allows the participating agents to learn their preferences via uniform allocation of services. Using the learnt estimates from this phase, the mechanism then runs \texttt{RRSD} in the second phase.

\subsection{Related Work}
We provide a brief discussion of the related work in the matching and bandit literature  
and highlight major differences comparing to our SBM models, which have not been considered previously. 




\textbf{Ordinal Matching and Distortion.}
We consider the objective of maximizing expected rewards as our offline benchmark. Since we do not observe the exact utilities of the agents rather ordinal preferences over items, we use the notion of \emph{distortion} \cite{PR06} from voting to quantify such a benchmark. In the context of voting, distortion measures the loss of performance due to limited availability of reward profiles \cite{BCHL+15, MPSW19, ABEP+18, Kempe20, AP17}. Our offline benchmark is related to the literature on the distortion of matching \cite{ABFV21, FFKZ14, AS16}. However, our offline benchmark needs to consider repeated matching over time, and because of the blocking of services, has a very different distortion than the distortion of a single-round matching.


\textbf{Online Matching.} 
There are existing online notions of weighted bipartite matching (e.g., \cite{Karp:1990un,Kalyanasundaram:1993wo,Karande:2011vz}) 
and stable matching (e.g., \cite{Khuller:1994vr})  
where the matching entities (i.e.\ agents or services) 
arrive dynamically over time and the corresponding information in the notions is publicly known 
(e.g., weights of the matched pairs or agents' ordinal preferences). 
These online settings are different from our repeated matching settings, where the entities 
do not arrive dynamically and our objective is to maximize expected rewards of the repeated matching 
given agents' ordinal preferences.  
Other recent works explore dynamic agent preferences that can change over time (e.g., \cite{Parkes:2013vs,Hosseini:2015tl,Hosseini:2015tc}). 
However, they do not consider the problem of maximizing expected rewards and blocking.  

\textbf{Blocking Bandits.} Our work in the online SBM models is closely related to the recent literature on blocking bandit models \cite{Basu:2021ue,Basu:2019ui,Bishop:2020un}, where each pulled arm (i.e., service)
can be blocked for a fixed number of rounds. Our work is also related to bandits with different types of arm-availability constraints \cite{NV14, KI18, KNS10}. 
However, these models do not consider the sequential matching setting where multiple strategic agents have (possibly unknown) ordinal preferences over arms  
and report ordinal preferences to a planner in order to be matched to some arm at each round. 



\textbf{Multi-agent multi-armed bandits.} The online setting in our work is broadly related to the growing literature on multi-agent multi-armed bandits \cite{LMJ20,SBS21,BBLB20}. \citet{LMJ20} consider a matching setting where strategic agents learn their preferences over time, and the planner outputs a matching every round based on their reported preferences. However, our setting is more challenging as we need to compete against a dynamic offline benchmark because of the blocking of services, whereas the existing works compete against a fixed benchmark e.g. repeated applications of Gale-Shapley matching in each round~\cite{LMJ20}.

\section{Preliminaries}
\label{sec:2}
 In this section, we introduce our model for sequential blocked matching. We start by describing how the preferences of each agent are modeled and describe formally how agents can be matched to services in a single time step. After which, we introduce the first of two sequential settings that we study in this paper, which features the blocking of services when they are assigned to agents.

    In our model, we have a set of agents, $N = \{1, \dots, n\}$, who hold \emph{cardinal} preferences over a set of services, $S = \{1, \dots, s \}$, where $s \gg n$ \footnote{Note that this is without loss of generality, as we may always add dummy services corresponding to a null assignment.}. We use $\mu_{i, j} \in \mathbb{R}^{+}$ to describe the cardinal reward agent $i$ receives for being assigned service $j$. Similarly, we denote by $\mu_{i} = (\mu)^{s}_{j=1}$ the vector of rewards associated with agent $i$. In what follows, we will also refer to $\mu_{i}$ as the \emph{reward profile} associated with agent $i$. Moreover, we restrict ourselves to reward profiles which lie in the probability simplex. That is, we assume $\mu_{i} \in \Delta^{s-1}$ for all $i \in N$. In other words, we make a  unit-sum assumption about the reward profile of each agent. Bounding constraints on reward profiles are common in the ordinal one-sided matching literature \cite{FFKZ14}, and are typically required in order to prove lower bounds for truthful algorithms such as RSD. Moreover, the unit-sum assumption is prevalent in social choice theory~\cite{BCHL+15}. Lastly, we denote by $\mu$ the $n$ by $s$ matrix of rewards.
    
    We say that agent $i$ (weakly) prefers service $a$ to service $b$ if agent $i$ receives greater reward by being assigned service $a$ over service $b$. That is, agent $i$ prefers service $a$ over service $b$ if and only if $\mu_{i, a} \geq \mu_{i, b}$. We use the standard notation $a \succ_{i} b$ to say that agent $i$ prefers service $a$ to service $b$. Additionally, we use the notation $\mathord{\succ}_{i}(j)$ to indicate the service in the $j$th position of the preference ordering $\succ_{i}$.  Note that every reward profile induces a linear preference ordering of services \footnote{One reward profile may induce many linear orderings. However, the linear preference profile induced by a reward profile can be made unique via tie-breaking rules.}. We use the notation $\mu_{i} \rhd \mathord{\succ_{i}}$ to denote that $\succ_{i}$ is a preference ordering induced by agent $i$'s reward profile. We let $\mathcal{P}(S)$, or $\mathcal{P}$ for short, denote the class of all linear preferences over $S$. We write $\succ_{i}$ to denote the preferences induced by agent $i$'s reward profile. Furthermore, we let $\succ = (\succ)_{i=1}^{n} \in  \mathcal{P}^{n}$ denote the \emph{preference profile} of the agents. As is standard, we write $\succ_{-i}$ to denote $(\succ_{1}, \dots, \succ_{i-1}, \succ_{i+1}, \dots, \succ_{n})$. As a result, we may denote $\succ$ by $(\succ_{i}, \succ_{-i})$.
    
    A \emph{matching} $m: N \to S\cup\{0\}$ is a mapping from agents to services. We let $m(i)$ denote the service allocated to agent $i$ by the matching $m$. We use $0$ to denote the null assignment. That is, agent $i$ is assigned no service in a matching if $m(i) = 0$. We let $\emptyset$ denote the null matching, in which no agent is assigned a service. We say matching is \emph{feasible} if no two agents are mapped to the same service. We let $\mathcal{M}$ denote the set of all feasible matchings. 
    
    In this paper, we consider discrete-time sequential decision problems, in which a planner  selects a sequence of (feasible) matchings over $T$ time steps. We let $m_{t}$ denote the matching chosen by the planner at time step $t$, and denote by $M = (m_{t})_{t=1}^{T}$ a sequence of $T$ matchings. We denote by $M(t, i) = m_{t}(i)$ the service matched to agent $i$ at time  $t$.
    
    Furthermore, we assume that, when a service is assigned, it may be blocked for a time period depending on the agent it was assigned to. More specifically, when agent $i$ is matched with service $j$, we assume that service $j$ cannot be matched to any agent for the next $D_{i, j} - 1$ time steps. We refer to $D_{i, j}$ as the \emph{blocking delay} associated with the agent-service pair $i$ and $j$. Additionally, we let $\tilde{D}$ denote the maximal blocking delay possible, and let $D$ denote the $n$ by $s$ matrix of blocking delays.
    
    From now on, we assume that all blocking delays are known a priori by both the planner and all agents. 
    We say that a matching sequence $M$ is \emph{feasible} with respect to the delay matrix $D$ if no service is matched to an agent on a time step where it has been blocked by a previous matching.
    \begin{defn}
    For a given blocking delay matrix $D$, the set of feasible matching sequences of length $T$, $\mathcal{M}^{D}_{T} \subseteq \mathcal{M}_{T}$, is the set of all matching sequences $M \in \mathcal{M}_{T}$ such that for all $t \in \{1, \dots, T\}$, $i \in N$, and $j \in S$, if $M(t, i) = j$ then $M(t^{\prime}, i^{\prime}) \neq j$ for all $i^{\prime} \in N$ and for all $t^{\prime}$ such that $t < t^{\prime} \leq t + D_{i, j} - 1$.
    \end{defn}
    In other words, we say that a matching sequence is feasible if there is no matching in the sequence which assigns an agent a service which has been blocked by a previous matching. Note that blocking of services is a common phenomenon in real-world scenarios. For example, consider a setting in which each service corresponds to a freelance contractor, and each agent corresponds to an employer. The matching of services and agents then corresponds to employers contracting freelancers. For the duration of the contract, which may differ from employer to employer, the matched freelancer is unavailable before returning to the pool of available services once their contract ends.
    
    We define the \emph{utility}, $\textrm{W}_{i}(M, \mu_{i})$, agent $i$ receives from a matching sequence $M$ as the sum of rewards it receives from each matching in the sequence. That is, $\textrm{W}_{i}(M, \mu_{i}) = \sum^{T}_{t=1}\mu_{i, M(t, i)}$. Similarly, we define the \emph{social welfare}, $\textrm{SW}(M, \mu)$, of a matching sequence $M$ as the summation of the utilities for all agents. More specifically, $\textrm{SW}(M, \mu) = \sum^{n}_{i=1}W_{i}(M, \mu_{i})$.  
    
    Next, we will describe the first sequential matching setting we consider in this paper, which we call the offline SBM setting. In this setting, the planner must produce a feasible matching sequence of length $T$. Prior to the selection of a matching sequence, each agent submits a linear preference ordering to the planner. We denote by $\tilde{\succ}_{i}$ the preference ordering, or \emph{report}, submitted by agent $i$. Analogously, we define $\tilde{\succ}$ as the preference profile submitted cumulatively by the agents, and call it the \emph{report profile}. A \emph{matching policy} $\pi(M\: |\: \tilde{\succ}, D)$ assigns a probability of returning a matching sequence $M$ given a submitted report profile $\tilde{\succ}$ and blocking delay matrix $D$. When it is clear from context, we will abuse notation and use $\pi(\tilde{\succ}, D)$ to refer to the (random) matching sequence prescribed by a policy $\pi$ given a report profile $\tilde{\succ}$ and blocking delay matrix $D$. 
    
    We say that a matching policy is \emph{admissible}, if for all possible report profiles and blocking delay matrices, the matching sequence returned by the policy is always feasible. The goal of the planner is to adopt an admissible matching policy which achieves high social welfare in expectation relative to the best feasible matching sequence in hindsight, $M^{*}(\mu, D) = \argmax_{M \in \mathcal{M}_{T}^{D}} \textrm{SW}(M, \mu)$.
    
    We assume that each agent, with full knowledge of the matching policy employed the planner, submits a linear preference ordering with the intention of maximising their own utility, and therefore may try to manipulate the planner by submitting a preference ordering which is not induced by their underlying cardinal preferences. We say that an agent is \emph{truthful} if they submit a preference ordering induced by their underlying cardinal preferences. That is, an agent is truthful if $\mu_{i} \rhd \tilde{\succ}_{i}$.  We denote by $\succ_{i}^{*}$ the report by agent $i$ which maximises agent $i$'s utility in expectation under the assumption that all other agents are truthful. We say that a policy is \emph{truthful} if for all possible $\mu$ and $D$ it is optimal for each agent to be truthful if all other agents are truthful. 
    In other words, a policy is truthful if for all $\mu$ and $D$ we have that $\mu_{i} \rhd \mathord{\succ}^{*}_{i}$ for all $i \in N$.
    
    To evaluate the efficiency of a given policy we use distortion, a standard notion of approximation for settings with ordinal preferences. 
    
    \begin{defn}
        The distortion of a matching policy is the worst-case ratio between the expected social welfare of the matching sequence, $\pi(\succ, D)$, returned by the policy under the assumption that all agents are truthful, and the social welfare of the optimal matching sequence, $M^{*}(\mu, D)$:
        \begin{equation*}
            \sup_{\mu, D}\frac{\textrm{SW}(M^{*}(\mu, D), \mu)}{\mathbb{E}\left[\textrm{SW}(\pi(\succ, D), \mu)\right]}
        \end{equation*}
    \end{defn}

    Note that distortion is the approximation ratio of the policy $\pi$ with respect to best matching sequence. In addition, note that the distortion is only a useful measure of a matching policies efficiency if said policy encourages truthful reporting. For example, for truthful policies, distortion is completely characterising of a policy's expected performance. As a result, we not only seek policies which have low distortion, but also policies which incentivise agents to submit their reports truthfully.
    
    To this end, we introduce the notion of incentive ratio, which measures the relative improvement in utility an agent can achieve by lying about their preferences.
    
    \begin{defn}
        The incentive ratio $\zeta(\pi) \in \mathbb{R}_{+}$ of a matching policy $\pi$ is given by:
        \begin{equation*}
            \zeta(\pi) = \max_{D, \mathord{\succ}_{-i},\: \mu_{i} \rhd \mathord{\succ}_{i}} \frac{\mathbb{E}[\textrm{W}_{i}(\pi((\succ_{i}, \succ_{-i}), D), \mu_{i})]}{\mathbb{E}[\textrm{W}_{i}(\pi((\succ^{*}_{i}, \succ_{-i}), D), \mu_{i})]}
        \end{equation*}
    \end{defn}
    

If a policy has an incentive ratio of $1$, then it is truthful. There are many reasons that we may expect a policy with bounded incentive ratio to do well. A bounded incentive ratio implies truth telling is a good approximation to the optimal report. If computing the optimal report is computationally intractable for the agent, being truthful is therefore an attractive alternative, especially if the approximation ratio implied by the incentive ratio is tight. In summary, we seek matching policies with good guarantees when it comes to both incentive ratio and distortion. This topic is treated in detail in the forthcoming sections.

\section{The Offline SBM Setting}\label{sec:offline-blockmatch}

In this section, we present our analysis of the offline SBM setting. We first provide a lower bound on the distortion achievable by both randomised and deterministic policies. Then, we discuss why trivial extensions of truthful one-shot matching algorithms do not result in truthful policies. Instead, we focus on designing policies which use truthful one-shot matching mechanisms as a basis, and have bounded incentive ratio. More precisely, we present the \texttt{RRSD} algorithm. We show that the incentive ratio of \texttt{RRSD} is bounded below by $ 1 - 1/e$, and provide upper bounds on the distortion achieved by \texttt{RRSD}, which match our previously established lower bounds on the best distortion achievable by any randomised algorithm. 

\subsection{Lower Bounds on the Distortion of Deterministic and Randomised Policies}
First, we prove that the distortion of any deterministic policy is $\Omega(s)$. That is, the distortion of any deterministic policy scales linearly with the number of services in the best case. In the proof, we first carefully construct a set of ordinal preferences. Then, given any matching sequence $M$, we show that there exists a set of reward profiles which induces the aforementioned ordinal preferences and on which $M$ incurs distortion of order $\Omega(s)$.\footnote{All missing proofs are deferred to the full version}


\begin{thm}
\label{thm: distortion of deterministic policies}
    The distortion of any deterministic policy is $\Omega(s)$.
\end{thm}

Next, we prove that the distortion incurred by any randomised policy is $\Omega(\sqrt{s})$. To prove this, we first show that it is sufficient to consider only \emph{anonymous} policies. That is, policies that assign each service to each agent the same number of times in expectation for all possible preference profiles. Then, we construct a set of reward profiles which yields the desired distortion for all anonymous truthful policies.


\begin{thm}
\label{thm:sqrt distortion}
    The distortion of the best randomised policy is $\Omega(\sqrt{s})$.
\end{thm}

\subsection{Constructing Truthful Algorithms for the Offline SBM Setting}
\label{sec:explain}
As previously mentioned, we assume that agents submit reports with the intention of maximising their own utility. As a result, the distortion incurred by a policy may not reflect its performance in practice, as agents may be incentivised to misreport their preferences in order to increase their utility. Note that in standard one-shot one-sided matching problems, this issue is sidestepped via the employment of truthful policies, like RSD. In addition, the restriction to considering truthful policies is well justified by the revelation principle. In a similar way, we would like to develop truthful algorithms for the offline SBM setting.

One may be tempted to apply such truthful one-shot policies to our setting directly. That is, to apply an algorithm such as RSD repeatedly on every time step in sequence in order to devise a matching sequence. This intuition is correct when there is no blocking, as the matching problems for each time step are then independent of each other. However, with blocking, the matchings from previous time steps will have a substantial effect on the set of matchings which preserve the feasibility of the matching sequence in future rounds. As a result, immediately obvious approaches, such as matching according to RSD repeatedly, do not result in truthful policies.

One simple way of generating truthful policies is to run a truthful one-shot one-sided matching policy once every $\tilde{D}$ time steps and simply return the empty matching in the remaining time steps. Such an approach decouples each time step from the next, resulting in truthfulness, but comes at the cost of only matching in at most $\lceil T/\tilde{D}\rceil$ rounds.

Instead, we construct an algorithm for the offline SBM setting from truthful one-shot matching algorithms in a different manner. More specifically, we propose the repeated random serial dictatorship (\texttt{RRSD}) algorithm, which uses RSD as a basis. Whilst \texttt{RRSD} is not truthful, it does have bounded incentive ratio.

\subsection{A Greedy Algorithm for the Offline SBM Setting}
The \texttt{RRSD} algorithm slowly builds up a matching sequence $M$ over time by iterating through agents and services. In other words, \texttt{RRSD} begins with the empty matching sequence, where $M(t, i) = 0$ for all $t$ and $i \in N$. To begin, \texttt{RRSD} samples a permutation of agents $\sigma$ uniformly at random. Next, \texttt{RRSD} iterates through the agents in the order given by the permutation sampled. For each agent $i$, \texttt{RRSD} iterates through services in the order specified by the preference ordering $\tilde{\succ}_{i}$ reported by agent $i$. For a given service $j$, \texttt{RRSD} repeatedly assigns service $j$ to agent $i$ at the earliest time step which does not cause the matching sequence to become infeasible. When no such time step exists, \texttt{RRSD}  moves onto the next service in agent $i$'s preference ordering. Once \text{RRSD} has iterated through the entire preference ordering of agent $i$, \texttt{RRSD} moves onto the next agent in the permutation $\sigma$ and repeats this process until the end of the permutation is reached.  The pseudocode for \texttt{RRSD} is given in the full version.

We will now briefly give the intuition behind \texttt{RRSD}. In essence, \texttt{RRSD} attempts to mimic the RSD algorithm for one-shot matching problems by allowing each agent to sequentially choose a feasible assignment of services over the entire time horizon (whilst respecting the assignments chosen by previous agents) via its reported ordering. In the case of RSD, given an agent's preference ordering, the same assignment is always optimal no matter the underlying the reward profile of the agent. That is, it is optimal for the agent to be assigned its most preferred available service, no matter its cardinal preferences. As a result, RSD is trivially truthful in the one-shot matching setting. In contrast, in the offline SBM setting, the optimal assignment of services can be different for two reward profiles which induce the same preference ordering. Hence, there is no trivial assignment, based on the preference ordering submitted by the agent which guarantees that agents are truthful.

Instead, given an agent's preference ordering, we attempt to find an assignment which performs reasonably well, no matter the underlying reward profile of the agent. \texttt{RRSD} uses a greedy algorithm to compute the assignment given to an agent. As long as this greedy algorithm is a good approximation of the optimal assignment, no matter the agent's underlying reward profile, then \texttt{RRSD} will have a bounded incentive ratio. The next theorem formalises this argument.

\begin{thm}
\label{thm: incentive ratio of RRSD}
    The incentive ratio of \texttt{RRSD} is asymptotically bounded below by $1 - 1/e$.
\end{thm}


\noindent
\textbf{Remark}. It is an open question as to whether we can achieve incentive ratios better than $1- 1/e$ when \texttt{RRSD} is used. In particular, one can show that many scheduling problems such as generic job interval scheduling and (dense) pinwheel scheduling can be reduced to the optimal manipulation problem each agent faces in \texttt{RRSD}. Whilst it is known that generic job interval scheduling problems are MAXSNP-hard~\cite{chuzhoy2006approximation}, it is still not known whether there exists a scheduling algorithm with approximation ratio better than $1 - 1/e$. 

We now provide an upper bound on the distortion achieved by \texttt{RRSD}, which matches our previously established lower bound for randomised policies described in Theorem \ref{thm:sqrt distortion}.

\begin{thm}
\label{thm:distortion of RRSD}
    The distortion of \texttt{RRSD} is at most $O(\sqrt{s})$.
\end{thm}

Finally, we show that it is possible to match the previously established lower bound for the distortion of deterministic algorithms. More specifically, we show that a derandomised version of \texttt{RRSD} incurs a distortion of at most $O(s)$. The main idea is that we can select $O(n^{2}\log{n})$ permutations of agents so that the number of times an agent $i$ is selected in the $j$th position is $\Theta(n\log{n})$. We can then run through these permutations one by one instead of selecting one permutation uniformly at random as in the \texttt{RRSD} algorithm. 

\begin{thm}
    \label{thm:det-upper}
    There is an admissible deterministic policy with distortion at most $O(s)$ for any $T \geq O(n^{2}\log(n))$
\end{thm}


\section{SBM with Bandit Feedback}
\label{sec:bandit}
Note that, in order for the guarantees above to hold in practice, we must assume that agents are fully aware of their ordinal preferences before matching begins. However, in many real-world scenarios, agents may be initially unaware of their preferences and learn them over time by matching with services. In addition, the reward an agent receives for being matched with a service may be inherently stochastic, depending on unobservable aspects of the underlying environment. With these concerns in mind we present a new sequential blocked matching setting, which we call the online SBM setting with bandit feedback, or online SBM for short.

In the online SBM setting, matching occurs time step by time step. At the beginning of each time step, every agent must submit a report, $\tilde{\succ}^{t}_{i}$, to the planner. The planner is then tasked with returning a matching of agents to services which obeys the blocking constraints imposed by the matchings from previous time steps. At the end of each time step, agent $i$ receives a reward, $r_{i,t} \in [0, 1]$, sampled from a distribution with mean $\mu_{i,j}$, where $j$ is the service agent $i$ was assigned in the matching returned by the planner. Additionally, we assume that each agent maintains an internal estimation, $\succ^{t}_{i}$, of its own preference ordering at every time step, based on the rewards received thus far.

We use $H^{\mathord{\succ}}_{t} = (\tilde{\mathord{\succ}}^{1}, \dots, \tilde{\mathord{\succ}}^{t})$ to denote the \emph{report history} up to time step $t$.
Furthermore, we use $H^{m}_{t} = (m_{1}, \dots, m_{t})$ to describe the \emph{matching history} at the end of time step $t$. We say that a matching history is \emph{feasible} if its matchings form a feasible matching sequence. Similarly, we use $H_{t}^{r} =(r_{1}, \dots, r_{t})$ to denote the \emph{reward history}. That, is the tuple of reward vectors, $r_{t}$, observed by the agents at every time step. An \emph{(online) matching policy} $\pi = (\pi_{1}, \dots, \pi_{T})$ is a tuple of functions $\pi_{t}(m | \tilde{H}^{\succ}_{t}, H^{m}_{t}, D)$ which assigns a probability of returning the matching $m$ given a report history $H^{\succ}_{t}$, a feasible matching history $H^{m}_{t}$ and a blocking delay matrix $D$. Similarly to the offline setting, we say that a matching policy is \emph{admissible} if it always returns a feasible matching sequence. 

Likewise, an \emph{(online) report policy} for agent $i$, $\tilde{\psi}_{i} = (\tilde{\psi}_{1}, \dots \tilde{\psi}_{t})$, is a tuple of functions $\tilde{\psi}_{t}(\tilde{\succ}_{i}^{t} | H^{r}_{t}, H^{m}_{t}, D)$ which assign a probability of agent $i$ reporting $\tilde{\succ}_{i}^{t}$ at time step $t$ given a reward history $H^{r}_{t}$, a matching history $H^{m}_{t}$, and blocking delay matrix $D$. We denote by $\tilde{\psi} = (\tilde{\psi}_{1}, \dots, \dots \tilde{\psi}_{n})$ the tuple of report policies used by the agents.  As before we use the notation $\tilde{\psi}_{-i}$ to denote the report policies of all agents bar agent $i$ and use $\psi$ to denote the tuple of report policies where each agent reports its internal estimation $\succ^{t}_{i}$ at every time step. We say that an agent is \emph{truthful} if it employs the report policy $\psi_{i}$.

The goal of each agent is to employ a report policy that maximises the sum of their own rewards across the time horizon. In contrast, goal of the planner is to employ a matching policy which maximises the sum of rewards across all agents and across all time steps.

In the bandit literature, a performance metric that is typically used to measure the efficiency of a policy is regret, which is defined as the expected difference between the rewards accumulated by a matching policy, and the expected reward accumulated by the best fixed matching policy in hindsight. That is, the best policy which repeatedly selects the same matching in as many time steps as possible. Such a benchmark policy may have very poor performance relative to the optimal matching sequence in expectation, and as such, the classical notion of regret is an unsuitable performance measure in the online SBM setting. To resolve this issue, we propose the following regret definition:
\begin{defn}
    The dynamic $\alpha$-regret of a policy $\pi$ is:
    \begin{equation*}
        R^{\alpha}_{\pi}(D, \mu, T) = \alpha\textrm{SW}(M^{*}, \mu) - \mathbb{E}_{\psi, \pi}\left[\sum^{n}_{i=1}\sum^{T}_{t=1}r_{i, t}\right] 
    \end{equation*}
\end{defn}
In other words, we compare the expected performance of a matching policy against a dynamic oracle which returns an  $1/\alpha$-optimal solution to the corresponding offline SBM problem, under the assumption that agents truthfully report their internal estimation of their preferences at each time step. Recall that, in the offline SBM setting, the distortion incurred by any policy is at least $\Omega(s)$. As a result, we cannot expect to construct algorithms with vanishing $1/\alpha$-regret for $\alpha < \sqrt{s}$.
In addition, one would not expect any matching policy to have low dynamic regret if the internal estimations computed by each agent are inaccurate. For example, if any agent's internal estimator consists of returning a random preference ordering, then we cannot hope to learn about said agent's preferences. As a result, we need to make reasonable assumptions regarding the internal estimator of each agent.

Similar to distortion for the offline SBM setting, dynamic $\alpha$-regret is only a meaningful performance measure for policies which motivate agents to adopt truthful reporting policies. Inspired by the concept of incentive ratio for the offline SBM setting, we define a new notion of regret which, given a matching policy $\pi$ captures the expected gain in cumulative reward an agent can achieve by misreporting.

\begin{defn}
    For a given matching policy $\pi$, we define agent $i$'s $\alpha$--IC regret (or $\alpha$ incentive compatible regret) as follows:
    \begin{equation*}
        \text{I}^{\alpha}_{\pi}(D, \mu, T) = \alpha \max_{\tilde{\psi}} \mathbb{E}_{\left(\psi_{-i}, \tilde{\psi}_{i}\right), \pi}\left[\sum^{T}_{t=1}r_{i, t}\right] -  \mathbb{E}_{\psi, \pi}\left[\sum^{T}_{t=1}r_{i, t}\right]
    \end{equation*}
\end{defn}
Note that for some matching policies, computing the optimal reporting policy may be computational intractable. If agents have vanishing $\alpha$-IC regret for a such  policy, then adopting a truthtelling  forms a good approximation of each agent's optimal reporting policy. If this approximation is better than what can be computed by the agent, then we can expect each agent to adopt their truthful reporting policy. Thus, we seek matching policies with good guarantees with respect to both dynamic $\alpha$-regret and $\alpha$-IC regret.

\subsection{Algorithms for Online SBM}
Next, we present a matching policy which achieves meaningful guarantees with respect to both dynamic $\alpha$-regret and $\alpha$-IC regret. More precisely, we present the bandit repeated random serial dictatorship (\texttt{BRRSD}) algorithm. Before we describe \texttt{BRRSD} formally, we first state our assumptions regarding the internal estimator used by each agent. 

Let $\hat{\mu}_{i, j}$ denote the empirical mean of the reward samples agent $i$ receives from being assigned service $j$. We say that an agent $i$ is \emph{mean-based} if service $a$ is preferred to service $b$ in $\succ^{t}_{i}$ if and only if $\hat{\mu}_{i, a} \geq \hat{\mu}_{i, b}$. That is, a mean-based agent  prefers services with higher empirical mean reward. From hereon, we assume that all agents are mean-based. 

Additionally, we use $\Delta_{\text{min}}$ to denote the smallest gap in mean rewards between two services for the same agent. That is, $\Delta_{\text{min}} = \min_{i,a\neq b}\left|\mu_{i,a} - \mu_{i, b}\right|$. Note that $\Delta_\text{min}$ is analogous to common complexity measures used in bandit exploration problems. Intuitively, if the mean rewards received from being assigned two services are similar, it will take more samples for a mean-based agent to decide which service they prefer. 

We are now ready to describe \texttt{BRRSD}. \texttt{BRRSD} is split into two phases. In the first phase, \texttt{BRRSD} assigns each agent each service exactly
$\left\lceil2\log(2Tsn)/ \Delta_{\text{min}}^{2}\right\rceil$ times. \texttt{BRRSD} performs these assignments in a greedy manner. At each time step, \texttt{BRRSD} iterates through the agent-service pairs that still need to be assigned in an arbitrary order. If an agent-service pair does not violate blocking constraints, then it is added to the current matching. Once this iteration is completed, or all agents have been assigned services, the matching is returned and \texttt{BRRSD} moves onto the next time step. Once all required assignments have been completed, \texttt{BRRSD} waits until all services are available, matching no agents to services in the meantime. Note that this takes a maximum of $\tilde{D}$ rounds. Then, \texttt{BRRSD} begins its second phase. At the beginning of the next time step, \texttt{BRRSD} observes the report profile $\tilde{\succ}_{i}^{t}$ and selects matchings according to \texttt{RRSD} using this report profile for the remainder of the time horizon. The full pseudocode for \texttt{BRRSD} is deferred to the full version.

\texttt{BRRSD} falls in the class of explore-then-commit (ETC) algorithms common in the bandit literature. The first phase of \texttt{BRRSD} serves as an exploration phase in which agent's learn their preference ordering. Meanwhile, the second phase of \texttt{BRRSD} serves as exploitation phase in which agents have the opportunity to disclose their accumulated knowledge to the planner in the form of ordinal preferences. Observe that this decoupling of exploration and exploitation avoids complicated incentive issues that may arise for sequential algorithms, which make no such clear separation. 

The exploration phase of \texttt{BRRSD} is simple relative to typical approaches in the bandit exploration literature. One may hope to apply a more complicated scheme for exploration, however approaches with better performance guarantees typically depend directly on the reward samples observed, which the planner does not access to.
The next theorem describes the  guarantees of \texttt{BRRSD}.

\begin{thm}
\label{thm:bandit}
Under the assumption that agents are mean-based, the following is true for all $\mu$ and $D$: \\
(i) The dynamic $(1/\sqrt{s})$-regret of \texttt{BRRSD} is $O\left(\tilde{D}\sqrt{s}\log\left(Tsn\right)/\Delta^{2}_{\text{min}}\right)$. \\
(ii) The $(1-1/e)$-IC regret for all agents under \texttt{BRRSD} is $O\left(\tilde{D}s\log\left(Tsn\right)/\Delta^{2}_{\text{min}}\right)$. \\
(iii) The greedy algorithm used by \texttt{BRRSD} in the exploration phase uses at most twice as many time steps as the shortest feasible matching sequence which completes the required assignments.
\end{thm}
\section{Conclusions and Future Work}
In this paper, we introduced the sequential blocking matching (SBM) model to capture repeated one-sided matching with blocking constraints. For the offline setting, we lower bounded the performance of both deterministic and randomised policies, presented algorithms with matching performance guarantees and bounded incentive ratio. Then, we analysed an online SBM setting, in which agents are initially unaware of their preferences and must learn them. For this setting, we presented an algorithm with sublinear regret with respect to an offline approximation oracle.

There are many interesting directions for future work. A natural generalisation would be to consider a two-sided matching setting \cite{Roth92} where services also hold preferences over agents. Additionally, our algorithms for both the offline and online settings are centralised. It is worth investigating whether similar performance guarantees can be achieved by a decentralised approach. Furthermore, we assumed that the preferences are static over time. It remains to be seen whether our approach generalises to settings where agents' preferences are dynamic and change over time \cite{BV19}.


\section*{Acknowledgements}
Nicholas Bishop was supported by the UK Engineering and Physical Sciences Research Council (EPSRC) Doctoral Training Partnership grant.

\printbibliography


\appendix

\section{Proofs}
\label{appendix - proofs of section 3}

\subsection{Proof of Theorem~\ref{thm: distortion of deterministic policies}}

\begin{proof}
    
    We now consider an instance of SBM with $n$ agents and $n$ services, where each agent has the same preferences. That is, $\mathord{\succ}_{a} = \mathord{\succ}_{b}$ for all $(a, b) \in N$. Furthermore, assume that, without loss of generality, service $j$ is in the $j$th position of this preference ordering. That is, assume that $\mathord{\succ_{a}}(j) = j$ for all $j \in S$. Lastly, assume that the blocking delay on each of the $n$ services is $\tilde{D}$ for all agents. In other words, put more formally, assume that $D_{ij} = \tilde{D}$ for all $i \in N$ and all $j \in S$.  
    
    We proceed with the proof in the following manner. Given the matching sequence $M$ returned by a deterministic policy using the above preference profile and blocking delay matrix, we will show that there exists a set of reward profiles which induce the preference profile, and on which the matching sequence $M$ suffers a distortion of $O(n) = O(s)$. We construct this reward profile via an inductive argument.
    
    
    Firstly, observe that there must exist some agent, $i_{1} \in N$, who is assigned service $1$ at most $T/\tilde{D}n$ times in the matching sequence $M$ by the pigeonhole principle. We set the reward profile of agent $i_{1}$ to $(1, 0, \dots, 0)$. Disregarding agent $i_{1}$, observe that there must exist different agent, $i_{2} \in N$, who is assigned service 1 or 2 at most $T/\tilde{D}(n-1)$ times, once again by the pigeonhole principle. We set the reward profile of agent $i_{2}$ to  $(1/2, 1/2, 0, \dots, 0)$. Disregarding both agents $i_{1}$ $i_{2}$, we can find a new agent, $i_{3} \in N$, who has been assigned services 1, 2 or 3 at most $T/\tilde{D}(n-2)$ times. we set the reward profile of agent $i_{3}$ to $ (1/3, 1/3, 1/3, 0 \dots, 0)$. We proceed in this pattern for a total of $n$ steps, until all agents are assigned reward profile. Note that the reward profiles induce the desired preference profile (assuming that a numeric tie-breaking rule is used).
    
    Given the assigned reward profiles, it is obvious that an optimal matching sequence assigns service $j$ to agent $i_{j}$ whenever the service is available. The social welfare of this optimal matching sequence is therefore of order $O(\log{(n)}T/\tilde{D})$. In contrast, the matching sequence $M$ has social welfare of order $O(\log{(n)}T/\tilde{D}n)$. As we can always construct such a reward profile, no matter the matching sequence $M$ returned by a policy, this implies that the distortion of any policy is of order $O(n) = O(s)$.
    
    
\end{proof}

\subsection{Proof of Theorem~\ref{thm:sqrt distortion}}

\begin{proof}
    Similarly to Theorem~\ref{thm: distortion of deterministic policies}, 
    we consider an instance of SBM with $n$ agents and $n$ services. Additionally, assume that the blocking delays for all services is the same for all agents. That is, $D_{ij} = d$ for some $d \leq \tilde{D}$. 
    
    Before moving to the content of the proof, we first show that it is sufficient to consider only \emph{anonymous} policies. Given a preference profile $\mathord{\succ}$, we let $A_{ij}(\mathord{\succ}) \in \set{0,1,\ldots,T}$ denote the random variable that indicates the number of times agent $i$ was allocated service $j$. We call a randomised matching algorithm anonymous if $\E[A_{ij}(\mathord{\succ}_1,\ldots,\mathord{\succ}_n)] = \E[A_{\sigma(i)j}(\mathord{\succ}_{\sigma(1)},\ldots, \mathord{\succ}_{\sigma(n)})]$ for all permutations $\sigma$. In other words, a matching policy is anonymous if each agent is assigned each policy the same number of times in expectation, regardless of the agents' relative positions in the preference profile.
    
    Now, suppose we are given a matching policy which has distortion at most $\rho$ i.e. $\sum_{ij} \mu_{ij} \E[A_{ij}(\mathord{\succ})] \ge \rho \opt(\mu)$. We can consider a new matching policy that selects a  permutation $\sigma$ uniformly at random and then applies the same policy on the input $\mathord{\succ}_\sigma = (\mathord{\succ}_{\sigma(1)}, \ldots, \mathord{\succ}_{\sigma(n)})$. Then the expected social welfare of the new policy is 
    $$
    \E_{\sigma}\left[\sum_{ij} \mu_{\sigma(i)j} A_{\sigma(i)j}(\mathord{\succ}_\sigma) \right] \ge \E_{\sigma}\left[ \rho \opt(\mu_\sigma)\right] = \rho \opt(\mu)
    $$
    The first inequality follows because the original policy gives $\rho$ distortion even when applied to the profile $\mu_\sigma$ and the second equality follows because the optimal welfare ($\opt(\mu) = \sum_{ij} \mu_{ij} A^*_{ij}$) is invariant to permutation. Therefore, the new anonymous policy has distortion at most $\rho$. This implies that for any matching policy, there is an anonymous matching policy with identical performance with respect to distortion. As a result, from now on, we restrict our consideration to anonymous matching policies without loss of generality.
    
    Next we will show that any anonymous matching policy incurs distortion of order $\Omega(s)$ via construction of a special set of reward profiles. The reward profiles we construct are very similar to the ones constructed in the proof of Lemma 8 of \cite{FFKZ14}. For each $i \in [\sqrt{n}]$, define 
    \begin{align*}
        \mu_{i, j} = \left\{\begin{array}{cc}
            1 - \sum_{j \neq i} \mu_{i, j} & \textrm{if } j = i\\
            \frac{n-j}{10 n^3 d} & \textrm{o.w.}
        \end{array} \right.
    \end{align*}
    And for each $\ell \in [\sqrt{n}-1]$, define
    \begin{align*}
        \mu_{i+\ell\sqrt{n}, j} = \left\{\begin{array}{cc}
             1 - \sum_{j \neq i} \mu_{i, j} & \textrm{if } j = i\\
             \frac{1}{\sqrt{n}} - \frac{j}{10n^2} & \textrm{if } j \neq i \And j \le \sqrt{n}\\
             \frac{n-j}{10 n^3 d} & \textrm{o.w.}
        \end{array}
        \right.
    \end{align*}
     The $n$ agents are grouped into $\sqrt{n}$ groups and all agents in group $i$ have the same preference order. Let $G_i = \set{i} \cup \set{i + \ell \sqrt{n}: \ell=1,\ldots,\sqrt{n}-1}$. Observe that all the agents in group $G_i$ have preference order $i \succ 1 \succ \ldots \succ i-1\succ i+1 \succ \ldots \succ n$. Therefore, for any service $j$, all the agents in group $G_i$ have  the same expected number of allocations. Let us call this number of allocations $T_{ij}$. Since any service $j$ can be allocated at most $T/d$ times we have
     \begin{equation}\label{eq:bound-sum-Tij}
     \sum_{i=1}^{\sqrt{n}} \sum_{p \in G_i} T_{ij} \le \frac{T}{d} \ \Rightarrow \sum_{i=1}^{\sqrt{n}} T_{ij} \le \frac{T}{d\sqrt{n}}
     \end{equation}
We now bound the expected social welfare of any randomized and anonymous matching policy with the given reward profile. For any agent $i \in [\sqrt{n}]$, the maximum expected utility over the $T$ rounds is at most $T_{ii} + \sum_{j \neq i} T_{ij} \frac{n-j}{10 n^3 d} \le T_{ii} + O\left( \frac{T}{nd}\right)$. Now consider an agent $i + \ell \sqrt{n}$ for $\ell \in [\sqrt{n}-1]$. Such an agent's utility over the $T$ rounds is at most $T_{ii}O\left(\frac{1}{\sqrt{n}} \right) + \sum_{j \neq i, j \le \sqrt{n}} T_{ij} \frac{1}{\sqrt{n}} + \sum_{j > \sqrt{n}} T_{ij} \frac{n-j}{10n^3 d} \le O\left(\frac{1}{\sqrt{n}} \right) \sum_{j=1}^{\sqrt{n}} T_{ij} + O\left(\frac{T}{nd} \right)$. Therefore, the total utility over all the $n$ agents is bounded by
\begin{align*}
&\sum_{i=1}^{\sqrt{n}}T_{ii} + O\left( \frac{1}{\sqrt{n}}\right) \sum_{i=1}^{\sqrt{n}}\sum_{\ell=1}^{\sqrt{n}-1} \sum_{j=1}^{\sqrt{n}} T_{ij} + O\left( \frac{T}{d}\right) \\
&\le \sum_{i=1}^{\sqrt{n}} T_{ii} + \sum_{i=1}^{\sqrt{n}} \sum_{j=1}^{\sqrt{n}} T_{ij} + O\left( \frac{T}{d}\right)  \\
&\le 2 \sum_{j=1}^{\sqrt{n}} \sum_{i=1}^{\sqrt{n}} T_{ij} +  O\left( \frac{T}{d}\right)  \\
& \le 2 \sum_{j=1}^{\sqrt{n}} \frac{T}{d\sqrt{n}} + O\left(\frac{T}{d}\right) = O\left(\frac{T}{d}\right)
\end{align*}
where the last line follows from equation \eqref{eq:bound-sum-Tij}. On the other hand, any deterministic and non-anonymous allocation rule that always assigns service $i$ to agent $i$ every $D$ rounds achieves a social welfare of at least $\sqrt{n}\frac{T}{d}(1 - \frac{1}{10nd}) \ge O\left(\frac{T \sqrt{n}}{d} \right)$. This establishes a bound of $O(\sqrt{n}) = O(\sqrt{s})$ on distortion.
\end{proof}

\subsection{Proof of Theorem~\ref{thm: incentive ratio of RRSD}}

\begin{proof}
    Without loss of generality, assume that agent $k$ is selected at random in the $k$th position of the permutation $\sigma$ sampled by $\texttt{RRSD}$. Assume, for the moment, that agents $1$ to $k-1$ are not allocated any services. Additionally, suppose that agent $k$ is free to choose its own allocation of services independent of the $\texttt{RRSD}$ algorithm. Under these assumptions, agent $k$ is posed with an offline blocking bandits problem as described in \cite{Basu:2019ui}. The solution proposed by $\texttt{RRSD}$ corresponds to a greedy approach in which the best service available is allocated at each time step. Thus, proving that such a  greedy algorithm has an approximation ratio of $1 - \frac{1}{e}$ implies the result in this restricted case. This fact was proven in \cite{Basu:2019ui}. We will show that this result holds more generally, regardless of the allocations chosen by $\texttt{RRSD}$ in previous time steps. 
    
    
    Again, assume agent $k$ is free to choose its own allocation, independent of $\texttt{RRSD}$. That is, agent $k$ is tasked with solving the following integer linear programming problem (ILP):
    \begin{align*}
        \max_{x_{t,j}}\quad & \sum^{T}_{t=1}\sum^{s}_{j=1}\mu_{k,j}x_{t, j}  \\
        \text{s.t. }\quad &  x_{t, j} \in \{0, 1\} &\quad \forall j \in S \\
        \quad & y_{t, j} + x_{t, j} \leq 1 &\quad \forall t \in [T], \forall j \in S \\
        &\sum^{s}_{j=1}x_{t, j} = 1 \quad &\forall t \in [T] \\
        &\sum_{t \in [D_{k, s}]}x_{t + t_{0}, j} \leq 1 &\quad \forall t_{0} \in T, \forall j \in S
    \end{align*}
    The variables $x_{t, j}$ indicate whether agent $k$ is assigned service $j$ at time step $t$. Meanwhile, the constants $y_{t, j}$ indicate whether agent $k$ cannot be assigned service $j$ on time step $t$ due to blocking delay constraints imposed by allocations of service $j$ to agents $1$ through $k-1$. The second set of constraints ensure that the assignments chosen by agent $k$ do not breach the delay constraints imposed by preexisting assignments of services to agents $1$ to $k-1$. The third set of constraints ensure that agent $k$ may only be matched to one service at each time step. Lastly, the fourth set of constraints ensure that agent $k$ chooses a sequence of assignments which obeys its own blocking delay constraints.
    
    We will proceed to develop an upper bound on this ILP through a series of relaxations. We will then compare this upper bound to an assignment which is outperformed by \texttt{RRSD} to prove an asymptotic bound on the incentive ratio, as desired. 
    
    \textbf{Computing the upper bound of the optimal solution}. We derive an upper bound for this ILP through a series of relaxations. First of all, we relax the integer constraints, so that at each time step agent $k$ can assign itself a fractional mixture of services. Additionally, we replace the constants $y_{t , j}$ with variables $z_{t, j}$ constrained to lie in $[0, 1]$. The idea in introducing these variables is to remove the blocking constraints imposed by the previous players and replace it with a constraint that stipulates that the total reduction in the time horizon available for agent $k$ to assign itself each service $j$ must remain the same. That is, agent $k$ is free to fractionally redistribute the blocked parts of the time horizon imposed by the previous $k-1$ agents. This results in the following linear program (LP):
    \begin{align*}
        \max_{x_{t,j},\: z_{t, j}}\quad & \sum^{T}_{t=1}\sum^{s}_{j=1}\mu_{k,j}x_{t, j}  \\
        \text{s.t. }\quad &  x_{t, j} \in [0, 1] \quad &\forall j \in S\\
        \quad &z_{t, j} + x_{t, j} \leq 1 \quad &\forall t \in [T], \forall j \in S \\
        &\sum^{s}_{j=1}x_{t, j} = 1 \quad &\forall t \in [T] \\
        &\sum_{t=1}^{T}z_{t, j} = \sum^{T}_{t=1}y_{t, j} \quad &\forall j \in [k-1]\\
        &\sum_{t \in [D_{k, s}]}x_{t + t_{0}, j} \leq 1 \quad &\forall t_{0} \in T, \forall j \in S
    \end{align*}
    It should be immediately obvious that this problem can be reformulated further, and, in fact, the individual fractional assignments per time step can be replaced with fractional assignments of agents to services for the entire time horizon. Similarly, the newly introduced auxiliary variables $z_{t, j}$ can be removed completely. In other words, it is clearly optimal for agent $k$ to spread the blocked parts of the time horizon evenly across all time slots, and then greedily match services to itself in each time step whilst obeying its own delay constraints. Therefore, it only matters how often each service is matched to agent $k$, as the fractional amount matched for every time step will be the same. This leads us to the following, equivalent, LP reformulation:
    \begin{align*}
        \max_{a_{j}}\quad & \sum^{s}_{j=1}a_{j}\mu_{k,j} \quad &\forall j \in S  \\
        \text{s.t. }\quad &  a_{j} \in [0, T/ D_{k, j}] \\
        \quad & a_{j} + \sum^{T}_{t=1}y_{t, j} \leq T \quad &\forall j \in S \\
        & \sum^{s}_{j=1}a_{j} = T
    \end{align*}
    Additionally, we define $C_{j} = \{t \in [T] \: : \: y_{t, j} = 1\}$ as the set of time steps in which agent $k$ cannot (in practice) be matched with service $j$ because of delay constraints imposed by previous agents. 
    Next, we show that this LP can be further formulated as a fractional bounded knapsack problem as follows.
    
    Consider each service $j$ as an item with weight $D_{k,j}$ and value $\mu_{k,j}$. From this perspective, $a_j$ is the (fractional) number of times we pack item $j$ into a knapsack (whose capacity is $T$). Note that the maximum value $a_j$ can get in the previous LP is determined by the pattern of $C_j$, and is also capped by $T/ D_{k, j}$.
    Therefore, in our bounded knapsack formulation, we can replace the constraints of $a_j$ to be $a_j \leq T/ D_{k, j} - b_j$ where $b_j$ is the number of blocks caused by $C_j$. Note that in general $b_j \neq |C_j|$, as it heavily depends on the pattern of the blocks. 
    Since $a_j \geq 0$, we have that $\frac{T}{D_{k,j}} \geq b_j$.
    It is well known that this fractional  bounded knapsack admits the optimal solution $\forall j \in S, a^{*}_{j} = \min \{T/D_{k,j} - b_j, (T - b_j - \sum_{l=1}^{j-1}a_{l}^{*})^{+}\}$. Note that the solution $a^{*}_{j}$ implicitly specifies an upper bound for the original ILP.
    
     \textbf{Computing the lower bound of the greedy sequence of matches}. Now consider the greedy sequence of matches for agent $k$ generated by the $\texttt{RRSD}$ algorithm. Let $a^{g}_{j}$ denote the number of times service $j$ is matched to agent $k$ by $\texttt{RRSD}$. Similarly let $A_{j}$ denote the set of time slots in which agent $k$ is allocated services 1 to $j-1$. The time slot where the periodic matching of service $j$ to agent $k$ collides with previous matches is denoted by $col_{j} = \{t \in A_{j}\cup C_{j} \: : \: D_{k, j}\;|\;t \}$. The number of times service $j$ is assigned to agent $k$ is at least $\lceil (T - |col_{j}|)/D_{k, j}\rceil$. This holds because for service $j$ we can remove the time slots with collisions and perform periodic placement perfectly with the remaining. $T - |col_{j}|$ time slots. Note that $|col_{j}| \leq \sum^{j-1}_{l=1}a^{g}_{l} + \sum^{T}_{t=1}y_{t, j} - |A_{j}\cap C_{j}|$.
    
    We now define for each $j \in S$, $a^{\prime}_{j} = T_{j}/D_{k, j} - b_j$,
    and $T_{j} = \left(T - \sum^{j-1}_{l=1}a^{\prime}_{l} + |A_{j}\cap C_{j}|\right)^{+}$. 
    We claim that $\sum^{j}_{l=1}a^{g}_{l} \geq \sum^{j}_{l=1}a^{\prime}_{l}$. In turn, this immediately implies that $\sum^{s}_{j=1}a_{j}^{g}\mu_{k, j} \geq \sum^{s}_{j=1}a_{j}^{\prime}\mu_{k, j}$. In other words, we will attempt to prove that the solution $a^{\prime}_{j}$ specifies a lower bound on the performance of \texttt{RRSD}. 
    
    We prove the claim using induction on $j$. We know that $a^{g}_{1} \geq \lceil (T - b_1)D_{1, k}\rceil$, so the base case is satisfied. By the inductive hypothesis, assume that $\sum_{l=1}^{j}a^{g}_{l} \geq \sum^{j}_{l=1}a^{\prime}_{l}$ for all $j < j^{\prime}$. We have:
    \begin{align*}
        &a^{g}_{j^{\prime}} \geq \lceil (T - |col_{j^{\prime}}|)/D_{k, j^{\prime}}\rceil \\
        &\geq \frac{1}{D_{k, j^{\prime}}}\left(T - \sum^{j^{\prime}-1}_{l=1}a^{g}_{l} - b_{j^{\prime}} + |A_{j^{\prime}}\cap C_{j^{\prime}}|\right) \\
        &= \frac{1}{D_{k, j^{\prime}}}\left(T - \sum^{j^{\prime}-1}_{l=1}a^{\prime}_{l} - \sum^{j^{\prime}-1}_{l=1}\left(a^{g}_{l} - a^{\prime}_{l}\right) - b_{j^{\prime}} + |A_{j^{\prime}}\cap C_{j^{\prime}}|\right) \\
        &=  a^{\prime}_{j^{\prime}} - \frac{1}{D_{k, j^{\prime}}}\sum^{j^{\prime}-1}_{l=1}\left(a^{g}_{l} - a^{\prime}_{l}\right)
    \end{align*}
    Thus we have that 
    \begin{equation*}
        \sum^{j'}_{l=1}(a^{g}_{l} - a^{\prime}_{l}) \geq (1 - 1/D_{k, j^{\prime}})\sum^{j^{\prime}}_{l=1}(a^{g}_{l} - a^{\prime}_{l})
    \end{equation*}
    which means $\sum^{j^{\prime}}_{l=1}a^{g}_{l} \geq \sum^{j^{\prime}}_{l=1}a^{\prime}_{l}$, and the inductive hypothesis holds. In what follows, we will refer to $a^{\prime}$ as the lower bound solution. Similarly, we will refer to $a^{*}$ as the upper bound solution.
    
 \textbf{Comparing the bounds}.  
Note that for any $j$, if $\frac{T}{D_{k,j}} = b_j$ then both the upper bound and lower bound solutions will not contain service $j$ (as $a^{\prime}_j \leq a^*_j = 0$).  
Therefore, without loss of generality, we assume that 
$\frac{T}{D_{k,j}} > b_j$. 
We set $D^{\prime}_{k,j}$ such that $\frac{1}{D^{\prime}_{k,j}} = \frac{1}{D_{k,j}} - \frac{b_j}{T}$. 
With induction in $j$ we can show that $a^{\prime}_j = \frac{T}{D^{\prime}_{k, j}}\prod_{l=1}^{j-1}(1-\frac{1}{D^{\prime}_{k, l}})$. 
In addition, we can also show that $a^*_j \leq \frac{T}{D^{\prime}_{k,j}} +1$.
The remainder of the proof consists of showing that the allocation $a^{\prime}_{j}$ performs asymptotically well compared to the optimal allocation $a^{*}$ via the closed forms above, which in turn implies the desired asymptotic bound on the incentive ratio of \texttt{RRSD}. The proof of this fact is contained in the proof for the blocking bandits setting provided by \cite{Basu:2019ui}, and as a result is omitted. We point the enthusiastic reader to the 'Greedy Lower Bound vs LP Upper Bound' subsection of the proof of Theorem 3.3 in \cite{Basu:2019ui}.

\end{proof}

\subsection{Proof of Theorem~\ref{thm:distortion of RRSD}}

\begin{proof}

Our proof proceeds by upper bounding the distortion of \texttt{RRSD} by the distortion of \texttt{RSD} on a new set of reward profiles. The distortion of $\texttt{RRSD}$ is given as
     $$
     \rho = \sup_{\mu,D} \frac{\textrm{SW}(M^{*}(\mu, D), \mu)}{\mathbb{E}[\textrm{SW}(\texttt{RRSD}(\mathord{\succ}, D), \mu)]}
     $$
     We first upper bound $\textrm{SW}(M^{*}(\mu, D), \mu)$ by expected welfare on a new instance with no blocking. For this new instance, there are $\tilde{D} n$ total agents and $s$ services. The $D_{\max} n$ agents are partitioned into $n$ groups, one for each agent in the original profile. We will write the group of agents corresponding to agent $i$ in the original profile by $G_i$ and the $\ell$-th agent in group $G_i$ will be denoted by $i_\ell$. The new reward profile $\tilde{\mu} \in \mathbb{R}^{D_{\max} n \times s}$ is defined as follows
     $$
     \tilde{\mu}_{i_\ell, j} = \frac{\mu_{i,j}}{D_{ij}} \quad \forall j \ \forall i_\ell \in G_i
     $$
     In the new instance, there is no blocking i.e. all the blocking lengths are one. 
     
     Now let $\pi^* = M^*(\mu,D)$ be the optimal policy for the original instance $\mu$ with blocking. We now construct a new policy for the non-blocking setting with reward matrix $\tilde{\mu}$. The new policy $\tilde{\pi}$ works as follows. At time $t$, if $\pi^*$ allocated service $j$  to agent $i$, then we select one available agent from the group $G_i$ (say $i_\ell$) and repeatedly allocate service $j$ to agent $i_\ell$ for the next $D_{i,j}$ rounds i.e. we set $\tilde{\pi}_{t'}(i_\ell) = j$ for $t'=t, t+1,\ldots, t+D_{i,j} - 1$. Notice that, since there are $\tilde{D}$ agents in group $G_i$, it is always possible to find such an available service $i_{\ell}$ whose allocation hasn't been determined at round $t$. This is because under the original policy $\pi^*$, at any time at most $\tilde{D}$ services can be simultaneously blocked as result of being assigned to agent $i$. Algorithm~\ref{alg:new-policy} describes how to construct the new policy $\tilde{\pi}$ from the old policy $\pi^*$.
     
     Let us now compare the social welfare of policy $\pi^*$ with reward instance $\mu$, and social welfare of policy $\tilde{\pi}$ with reward instance $\tilde{\mu}$. Notice that whenever $\pi^*$ allocates service $j$ to agent $i$, a corresponding agent (say $i_\ell$) is assigned service $j$ exactly $D_{i,j}$ times under the new policy $\tilde{\pi}$. As the new rewards are normalized by the blocking lengths, this implies that the total reward gathered by $i$ under $\pi^*$ is the same as the total reward gathered by all the agents in $G_i$ under the new policy $\tilde{\pi}$. Thus, summing over all the agents, we have $\SW(\pi^*, \mu) = \SW(\tilde{\pi}, \tilde{\mu})$. 
     
     Now observe that, under the new instance $\tilde{\mu}$, there is no blocking, so the optimal allocation rule is obtained by applying a fixed matching (say $\sigma^*$)\footnote{Ideally $\sigma^*$ is an assignment from $D_{\max} n$ agents to $s$ services based on the optimal achievable social welfare and not a one-to-one matching. But we will use the term matching instead of assignment to be consistent with the rest of the paper.} repeatedly over the $T$ rounds. Let $\SW_0(\sigma^*, \tilde{\mu})$ be the one-round social welfare of the matching $\sigma^*$. Then under the non-blocking reward instance $\tilde{\mu}$, the best possible social welfare is $T \cdot \SW_0(\sigma^*, \tilde{\mu})$. This gives us the following bound on the welfare of the original policy $\pi^*$ for the blocking instance.
     \begin{equation}
         \label{eq:ubd-orig-policy}
         \SW(\pi^*, \mu) \le \SW(\tilde{\pi}, \tilde{\mu}) \le T \cdot \SW_0(\sigma^*, \tilde{\mu})
     \end{equation}
     
     \begin{algorithm}[!t]
     \SetAlgoLined
\KwInput{ $T$, $N$, $S$, $D$, Policy $\pi^*$}
\KwOutput{ $\tilde{\pi}$ }
\tcc{Matrix $F$ keeps track of the available agents within each group $G_i$.}
$F(i,i_\ell) = 0 \ \forall i_\ell \in G_i\ \forall i \in N$ \\
 \For{$t \in \set{1,\ldots,T}$}{
    \For{$i \in N$}{
        \If{$\pi^*_t(i) = j$}{
            \tcc{Find an available agent within group $G_i$}
            Choose $i_\ell$ s.t. $F(i,i_\ell) = 0$\\
            $\tilde{\pi}_{t'}(i_\ell) = j \ \forall t' \in [t,\ldots,t+D_{ij} - 1]$\\
            $F(i,i_\ell) = 1$\\
        }
        \tcc{If any arm becomes available under $\pi^*$, then we make corresponding agents available}
        \For{$j \in S$}
        {
            \If{$\pi^*_{t-D_{i,j} + 1}(i) = j$}
            {
                \tcc{$i_j \in G_i$ is the corresponding agent with repeated allocations from $\set{t-D_{i,j}+1,\ldots,t}$ under $\tilde{\pi}$}
                $F(i,i_j) = 0$\\
            }
        }
    }
    
 }
\Return $\tilde{\pi}$\\
 \caption{Policy Conversion ($\pi^* \rightarrow \tilde{\pi}$).\label{alg:new-policy}}
     \end{algorithm}

       We now prove a lower bound on the expected social welfare of \texttt{RRSD} under the original reward instance $\mu$.  \texttt{RRSD} is a randomized policy, but for a given order of agents, the sequence of assignments generated by \texttt{RRSD} becomes a deterministic policy. Any such deterministic policy $\pi$ can be converted to an equivalent policy (say $g(\pi)$) through Algorithm~\ref{alg:new-policy}. Moreover, the new policy $g(\pi)$ preserves the social welfare under the new reward instance $\tilde{\mu}$. The expected social welfare of \texttt{RRSD} is given as 
       \begin{align*}
       &\E[\SW(\texttt{RRSD}, \mu)] = \E_{\pi \sim \texttt{RRSD}}[\SW(\pi, \mu)] \\&= \E_{\pi \sim \texttt{RRSD}}[\SW(g(\pi), \tilde{\mu})]
       \end{align*}
       The last line follows from the welfare preservation property of Algorithm~\ref{alg:new-policy}. Moreover, given a policy $\pi$, the new policy $g(\pi)$ actually assigns the same service repeatedly to the same agent. This is because whenever the original \texttt{RRSD} assigns a new service to agent $i$, algorithm~\ref{alg:new-policy} selects a new agent in $G_i$, and assigns the new service to the new agent. Now given $g(\pi)$ consider a simpler policy $g'(\pi)$ which only makes the first repeated assignments to any member from $G_i$ for all $i$. Since this new policy makes fewer assignments than $g(\pi)$ we have the following inequality.
       $$
       \E_{\pi \sim \texttt{RRSD}}[\SW(g(\pi), \tilde{\mu})] \ge \E_{\pi \sim \texttt{RRSD}}[\SW(g'(\pi), \tilde{\mu})]
       $$
       Since the sequence $\pi$ was generated under $\texttt{RRSD}$ for the original instance $\mu$, an alternative way to generate the sequence $g'(\pi)$ is the following: first random choose an order of the set $N$, and then replace each agent $i$ in this sequence with a randomly choose agent from $G_i$. Let us call this new randomised policy \texttt{GRSD} (short for grouped \texttt{RSD}). Then we have,
       \begin{align*}
       \E_{\pi \sim \texttt{RRSD}}[\SW(g'(\pi), \tilde{\mu})] &= T \cdot \E[\SW_0(\texttt{GRSD}, \tilde{\mu})] \\
       &= T \cdot \E[\SW_0(\texttt{RSD}, \tilde{\mu})] 
       \end{align*}
       The last equality follows from the following two observations. Under \texttt{GRSD}, the probability that an agent from group $i$ shows up at position $j$ equals $1/n$. On the other hand, under \texttt{RSD}, the probability that an agent from group $i$ shows up at position $j$ equals $\tilde{D} \times (1/\tilde{D} n) = 1/n$. Second, conditioned on the event an agent from group $i$ shows up at position $j$, the expected utility of the agent is the same, as all the agents in group $G_i$ are duplicates of the original agent $i$ and have the same reward profile. Therefore, we have established the following lower bound on the expected welfare of \texttt{RRSD} under the original instance $\mu$.
       \begin{equation}\label{eq:lbd-rrsd}
       \E[\SW(\texttt{RRSD}, \mu)]  \ge T \cdot \E[\SW_0(\texttt{RSD}, \tilde{\mu})]
       \end{equation}

    We can now bound the distortion of \texttt{RRSD} as follows.
    \begin{align*}
        \rho &= \sup_\mu \frac{\textrm{SW}(\pi^\star, \mu)}{\mathbb{E}_{\pi \sim \texttt{RRSD}}[\textrm{SW}(\pi, \mu)]} \\&\le \sup_{\tilde{\mu}} \frac{T \cdot \textrm{SW}_0(\sigma^\star, \tilde{\mu})}{T\cdot \mathbb{E}_{\sigma \sim \texttt{RSD}}[\textrm{SW}_0(\sigma, \tilde{\mu})]} \\&= \sup_{\tilde{\mu}} \frac{\textrm{SW}_0(\sigma^\star, \tilde{\mu})}{\mathbb{E}_{\sigma \sim \texttt{RSD}}[\textrm{SW}_0(\sigma, \tilde{\mu})]} 
    \end{align*}
    Where the inequality is due to the lower bound from  eq. \ref{eq:lbd-rrsd} and the upper bound from eq. \ref{eq:ubd-orig-policy}.
    Since the last quantity is just the distortion of \texttt{RSD} in the one-shot matching setting, we can apply Lemma 4 from \cite{FFKZ14} and get a bound of $O(\sqrt{s})$ on the distortion.\footnote{\citet{FFKZ14} actually considered a setting where $n=s$, but their proof naturally generalizes for the setting with $n > s$}
\end{proof}

\subsection{Derandomised \texttt{RRSD}}
In this section, we present a deterministic matching policy for the offline SBM setting. More precisely, we present derandomised \texttt{RRSD} (\texttt{DRRSD}), which, as the name suggests, is a derandomised version of \texttt{RRSD}. Instead of sampling a single permutation, like \texttt{RRSD}, \texttt{DRRSD} uses a set of $4n^{2}\log(n)$ permutations. In addition, this set is constrained to ensure that the fraction of permutations in which an agent $i$ appears in the $j$th position is at least $\frac{1}{2n}$. The following lemma stipulates that such a set of permutations always exists.

\begin{lemma}\label{lem:derandomization}
There exists a set of $4n^2 \log(n)$ permutations over $n$ agents such that the fraction of times agent $i$ appears at the $j$th position is at least $\frac{1}{2n}$.
\end{lemma}
\begin{proof}
    The proof is by the probabilistic method. Let us draw $P$ permutations over the $n$ agents uniformly at random. Let $X_{ij}$ be the fraction of times agent $i$ appears at $j$th position over the $P$ permutations. Then $\E[X_{ij}] = 1/n$. Moreover, from the Chernoff-Hoeffding inequality,
    $$
    P\left(X_{ij} \le \frac{1}{2n} \right) \le 2 e^{-2P\frac{1}{4n^2}} = 2 e^{-\frac{P}{2n^2}}.
    $$
    Moreover, by a union bound over the $n$ agents and $n$ positions we get that $$
    P\left(\exists i,j \ X_{ij} \le \frac{1}{2n} \right) \le 2 e^{-\frac{P}{2n^2}}.
    $$
    Therefore, if $P \ge 4n^2 \log(n)$, the probability of observing a set of permutations such that each $X_{ij} \ge 1/2n$ is positive. This implies that if $P = 4n^2 \log(n)$, we can find a required set of permutations.
\end{proof}
\texttt{DRRSD} splits the time horizon into evenly sized blocks. In each block, a different permutation is used. Within each block, agents are assigned to services by the same greedy method used by \texttt{RRSD}, with one caveat. If the blocking delay caused by the assignment of an agent-service pair would overrun into the next block, then this assignment is skipped. This to ensure that all services will be available at the beginning of each block. The pseudocode for \texttt{DRRSD} is presented in Algorithm \ref{alg:DRRDS}.

Next, we prove that \texttt{DRRSD} incurs a distortion of order $O(s)$, which matches the lower bound we established for the distortion of deterministic policies in Theorem \ref{thm: distortion of deterministic policies}.

\subsection{Proof of Theorem \ref{thm:det-upper}}
\begin{proof}
We will write $i_j$ to denote agent $i$'s $j$th favourite service. That is, $i_{j} = \mathord\succ_{i}(j)$. By Lemma \ref{lem:derandomization}, agent $i$ gets her $j$th favourite service (or better) in at least $\frac{P}{2n}$ groups. Within any such group,  there are $T/P$ time slots, and agent $i$ is assigned her $j$th favourite service at least $\left \lfloor T/(PD_{i,i_j}) \right \rfloor$ times. Therefore, the total welfare guaranteed by \texttt{DRRSD} is at least
\begin{align*}
&\sum_{i=1}^n \sum_{j=1}^S \mu_{i,i_j} \frac{P}{2n} \left \lfloor \frac{T}{PD_{i,i_j}} \right \rfloor\\ \ge &\sum_{i=1}^n \sum_{j=1}^s \mu_{i, i_j}\frac{P}{4n}  \frac{T}{PD_{i,i_j}} \\ &= \frac{T}{4n} \sum_{i=1}^n \sum_{j=1}^s \frac{\mu_{i, i_j}}{D_{i,i_j}}
\end{align*}

On the other hand, consider a matching algorithm that assigns service $i_j$ to agent $i$ exactly $A_{i,j}$ times. Whenever item $j$ is matched to agent $i$, it is blocked for $D_{i,i_j}$ rounds. This implies that $A_{i,j} \le T/D_{i,i_j}$. Therefore, the maximum welfare achievable by such a matching algorithm is at most
    \begin{align*}
        \sum_{i=1}^n \sum_{j=1}^s \mu_{i, j} A_{i,j} \le \sum_{i=1}^n \sum_{j=1}^s \mu_{i,j} \frac{T}{D_{i,i_j}}
    \end{align*}
    This establishes that the distortion of \texttt{DRRSD} is at most $4n = O(s)$.
\end{proof}

\begin{algorithm}[t!] 
\SetAlgoLined
\KwInput{ $T$, $N$, $D$, $S$, $\mathord{\succ}$, and a set of $P=4n^2\log(n)$ permutations $\set{\sigma_1,\ldots, \sigma_P}$}
 $M = (m_{t})^{T}_{t=1} = (\emptyset)^{T}_{t=1}$\\
 \For{$p = 1, \dots, P$}{
 $\sigma = \sigma_{p}$ \\
 \For{$i = 1, \dots, n$}{
    \tcp{Select agent}
    $\text{ag} = \sigma(i)$ \\
    $\text{start} = (p-1)T/P$ \\
    $\text{end} = pT/P$ \\
    \For{$j = 1, \dots, s$}{
        \tcp{Select service}
        $\text{ser} = \tilde{\succ}_{\text{ag}}(j)$ \\
        \While{$\text{available}(M,\text{ag}, \text{ser}, \text{start}, \text{end})$}{
           $t = \text{earliest}(M, \text{ag}, \text{ser}, \text{start}, \text{end})$ \\
           \If{$\text{overrun}(\text{ag}, \text{ser}, t, \text{end})$}{\textbf{break}}{
           $M(t, \text{ag}) = \text{ser}$\\
         }
       }
    }
}
 }
\Return $M$
 \caption{\texttt{DRRSD} (Derandomized \texttt{RRSD})\label{alg:DRRDS}}
\end{algorithm}

\begin{algorithm}[t!] 
\SetAlgoLined
\KwInput{ $T$, $N$, $D$, $S$, $\Delta_{\text{min}}$}
$M = (m_{t})^{T}_{t=1} = (\emptyset)^{T}_{t=1}$\\
$\text{repeats} = \left\lceil2\log(2Tsn)/ \Delta_{\text{min}}^{2}\right\rceil$ \\
\tcp{Build a list of exploration assignments}
$\text{jobList} = \text{buildJobList}(n, s, \text{repeats})$ \\
$\text{explore} = \text{true}$ \\
$\text{waiting} = \text{false}$ \\
$\text{count} = 0$ \\
\For{$t= 1, \dots T$}{
    \tcp{Exploration phase}
    \If{$\text{explore}$}{
        \tcp{Greedily add remaining agent-service pairs to current matching}
        \For{$(i, j) \text{ in jobList}$}
        {   
            \If{\text{available}($M$, $i$, $j$, $t$)}
                {$M(t, i) = j$ \\
                $\text{jobList.remove}(i, j)$}
        }
        
        \tcp{Start waiting phase}
        \If{$\text{jobList.isEmpty}$}{
            $\text{explore} = \text{false}$\\
            $\text{waiting} = \text{true}$}
    }
    \tcp{Wait for all services to become available}
    \If{\text{waiting}}{
        \If{$\text{count} < \tilde{D}$}{
            $\text{count}\mathord{+}\mathord{+}$
        }
        
        \tcp{Start exploitation phase}
        \Else{
            $\tilde{\succ} = \tilde{\succ}^{t}$ \\
            Sample $\sigma$ \\
            $\text{waiting} = \text{false}$\\
        }
    }
    \tcp{Exploitation phase}
    \Else{
        \For{$i = 1, \dots n$}{
            $\text{ag} = \sigma(i)$ \\
            \For{$j = 1, \dots, s$}{
                $\text{ser} = \tilde{\succ}_{\text{ag}}(j)$ \\
                \If{$\text{available}(M, \text{ag}, \text{ser}, t)$}{
                    $M(t, \text{ag}) = \text{ser}$
                }
            }
        }
    }
}

\Return $M$
 \caption{\texttt{BRRSD} (Bandit \texttt{RRSD})\label{alg:BRRSD}}
\end{algorithm}

\subsection{Proof of Theorem~\ref{thm:bandit}}

The pseudocode for \texttt{BRRSD} is given in Algorithm \ref{alg:BRRSD}. We now restate the three claims of the Theorem \ref{thm:bandit}, and proceed with a proof below: \\
\emph{
(i) The dynamic $(1/\sqrt{s})$-regret of \texttt{BRRSD} is $O\left(\tilde{D}\sqrt{s}\log\left(Tsn\right)/\Delta^{2}_{\text{min}}\right)$. \\
(ii) The $(1-1/e)$-IC regret for all agents under \texttt{BRRSD} is $O\left(\tilde{D}s\log\left(Tsn\right)/\Delta^{2}_{\text{min}}\right)$. \\
(iii) The greedy algorithm used by \texttt{BRRSD} in the exploration phase uses at most twice as many time steps as the shortest feasible matching sequence which completes the required assignments.}

\begin{proof}
Claim (i) can be proved as follows. By the end of the exploration phase, we know that each agent has received a reward from being assigned each service at least $\left\lceil2\log(2Tsn)/\Delta_{\text{min}}^{2}\right\rceil$ times. In addition, note that this exploration phase takes at most $\tilde{D}s\left\lceil2\log(2Tsn)/\Delta_{\text{min}}^{2}\right\rceil + \tilde{D}$ rounds. By the Chernoff-Hoeffding inequality,  we have that for all agents $i$ and services $j$:
\begin{equation*}
    P\left(\left|\mu_{i, j} - \hat{\mu}_{i, j}\right| \geq \frac{\Delta_{\text{min}}}{2} \right) \leq \frac{1}{Tsn}
\end{equation*}

Thus, by the union bound and the assumption that all agents are mean-based, with probability $1 - 1/T$, the internal estimation of every agent will be correct. From now, unless explicitly stated, we will assume that all agents have learned the correct preference ordering by the end of the exploration phase. 

For the sake of simplicity, let $T_{1}$ denote the number of rounds for which the exploration phase runs, and let $T_{2}$ denote the number of rounds for which the exploitation phase runs. Similarly let $\textrm{OPT}_{1}$ denote the social welfare of the optimal matching sequence of length $T_{1}$, and $\textrm{OPT}_{2}$ denote the social welfare of the optimal matching sequence of length $T_{2}$. Furthermore, let $\textrm{SW}_{1}(\texttt{BRRSD})$ denote the social welfare generated by $\texttt{BRRSD}$ in the exploration phase, and $\textrm{SW}_{2}(\texttt{BRRSD})$ denote the social welfare generated by $\texttt{BRSSD}$ in the exploitation phase.

As each reward is bounded between $[0, 1]$, and the exploration phase proceeds for at most $\tilde{D}s\left\lceil2\log(2Tsn)/\Delta_{\text{min}}^{2}\right\rceil + \tilde{D}$ time steps, it is easy to show that
\begin{align*}
    \frac{1}{\sqrt{s}}\mathbb{E}\left[\textrm{OPT}_{1}\right] - \mathbb{E}\left[\textrm{SW}_{1}(\texttt{BRRSD})\right]  \leq \\ \frac{1}{\sqrt{s}}\left(\tilde{D}s\left\lceil2\log(2Tsn)/\Delta_{\text{min}}^{2}\right\rceil + \tilde{D}\right)
\end{align*}
In addition, by Theorem \ref{thm:distortion of RRSD}, and the assumption that agents are mean-based, we have the following lower bound:
\begin{equation*}
    \frac{1}{\sqrt{s}}\mathbb{E}\left[\textrm{OPT}_{2}\right] \leq \mathbb{E}\left[\textrm{SW}_{2}(\texttt{BRRSD})\right]
\end{equation*}
Let $\textrm{OPT}$ denote the social welfare of the optimal matching sequence of length $T$. Combining the bounds above and noting that $\textrm{OPT} \leq \mathbb{E}[\textrm{OPT}_{1}] + \mathbb{E}[\textrm{OPT}_{2}]$ we have:
\begin{align*}
    \frac{1}{\sqrt{s}}\mathbb{E}[\textrm{OPT}]- \mathbb{E}[\textrm{SW}(\texttt{BRRSD})] \leq \\
    \frac{1}{\sqrt{s}}\left(\tilde{D}s\left\lceil2\log(2Tsn)/\Delta_{\text{min}}^{2}\right\rceil + \tilde{D}\right)
\end{align*}
For the case when at least 1 agent does not learn the correct preference ordering by the end of the exploration phase, the $(\frac{1}{\sqrt{s}})$-dynamic regret of $\texttt{BRRSD}$ is bounded above by $\frac{1}{\sqrt{s}}nT$. Combining both cases together, we see that the dynamic $(\frac{1}{\sqrt{s}})$-regret of $\texttt{BRRSD}$ is bounded above by $\frac{1}{\sqrt{s}}\left(\tilde{D}s\left\lceil2\log(2Tsn)/\Delta_{\text{min}}^{2}\right\rceil + \tilde{D}\right) + \frac{1}{\sqrt{s}}n$, implying the desired regret bound. 

To prove claim (ii), note that the an agent can only affect its assignment of services in the exploitation phase, in which $\texttt{BRRSD}$ deploys the $\texttt{RRSD}$ algorithm. Thus, following a similar argument as above, replacing the use of Theorem \ref{thm:distortion of RRSD} with Theorem \ref{thm: incentive ratio of RRSD}, we find that the $(1- 1/e)$-IC regret of $\texttt{BRRSD}$ is  bounded above by $(1 - 1/e)\left(\tilde{D}s\left\lceil2\log(2Tsn)/\Delta_{\text{min}}^{2}\right\rceil + \tilde{D}\right) + (1-1/e)n$.

Finally, to prove claim (iii), we consider the following scheduling problem:

\textbf{Open Shop Scheduling:} An instance of the open shop problem consists of a set of $N$ machines and $S$ jobs. Associated with each job $j$ is a set of $n$ independent tasks $j_{1}, \dots, j_{n}$. The task $j$ for job $i$ must be processed on machine $i$ for an uninterrupted $D_{i, j}$ time units. A schedule assigns every task $j_{i}$ to a time interval $D_{i, j}$ so that no job is simultaneously processed on two different machines, and so that no machine simultaneously processes two different jobs. The makespan $C_{max}$ of a schedule is the longest job completion time. The optimal makespan is denoted by $C^{*}_{max}$.

It is easy to show that the exploration phase of $\texttt{BRRSD}$ reduces to an open shop scheduling problem in which there is a job for each agent $i$, and a task for each assignment of service $j$ to agent $i$. Similarly, observe that the assignment procedure used by $\texttt{BRRSD}$ is simply an implementation of the greedy algorithm for open shop scheduling as described by ~\cite{woeginger2018open}. The claim follows from that fact that the greedy algorithm is a 2-approximation for open shop scheduling (see \cite{woeginger2018open}).

\end{proof}


\section{Computational complexity of Offline SBM}

In this section, we investigate the computational complexity of the offline SBM setting. Observe that, when there is only one agent, offline SBM corresponds to the stochastic blocking bandit problem defined in \cite{Basu:2019ui}. Therefore, offline SBM inherits the complexity issues of the offline stochastic blocking bandit problem. More precisely, there is no pseudopolynomial time algorithm for offline SBM setting unless the randomised exponential time hypothesis \cite{exp} is false.
\begin{thm}
     Offline SBM problem does not admit a pseudopolynomial time policy unless the randomised exponential time hypothesis is false.
\end{thm}
\section{Failure of Trivial Approaches}
In this section, we provide a more detailed exposition regarding the failure of trivial approaches in the SBM setting. As discussed in Section \ref{sec:explain}, one may be tempted to employ a one-shot matching algorithm on each time step, such as RSD. However, such an approach does not result in a truthful algorithm. Consider a problem instance involving three agents, $\{1, 2, 3\}$ and three services, $\{a, b , c\}$, in which the ordinal preferences of each agent are defined as follows:
\begin{align*}
1:\: a \succ b \succ c \\
2:\: b \succ c \succ a \\
3:\: b \succ c \succ a
\end{align*}
and the delay matrix is given by
\begin{equation*}
    D = \begin{pmatrix}
        2 & 1 & 1 \\
        2 & 2 & 2 \\
        1 & 1 & 1
    \end{pmatrix}.
\end{equation*}
Now, consider an algorithm in which a permutation $\sigma$ of the agents is sampled uniformly at random, and RSD is run repeatedly on each time step in accordance with $\sigma$. In particular, consider the case where  $\sigma = (1, 2, 3)$. Table \ref{tab:true} shows the assignment of services when each agent reports their true preferences to the algorithm. Note that agent $1$ receives service $a$, followed by service $c$, repeatedly for the entire time horizon. Meanwhile, Table \ref{tab:lie} shows the assignment of services when agent $1$ lies and reports the preference ordering $b \succ a \succ c$. Note that in this case, agent $1$ receives service $b$ followed by service $a$, repeatedly for the entire time horizon.  Since agent $1$ prefers service $b$ over service $c$, agent $1$ prefers the assignment of services it receives from misreporting when $\sigma = (1, 2, 3)$. Moreover, it is easy to verify that misreporting $b \succ a \succ c$, does not worsen agent $1$'s welfare when any other permutation is sampled. Therefore, we conclude that agent $1$ can benefit by misreporting the preference ordering $b \succ a \succ c$. Hence, we conclude that the algorithm proposed is not truthful.

Observe that this example illustrates SBM's close connection to scheduling problems. If agent $1$ takes service $a$ on the first round, then service $b$ will be blocked and unavailable in the next. As a result, it is better for agent $1$ to secure service $b$ as quickly as possible before it becomes blocked by other agents, and pursue service $a$ later. Intuitively, assuming that all other agents are truthful, one can view an agent's report as selecting a certain (potentially randomised)  schedule of service assignments. For an algorithm to be truthful, truthful reporting must return the best (expected) schedule of assignments for the agent out of those offered by the algorithm, given the preferences of other agents. For an algorithm to be optimal in terms of distortion, the agent needs a wide range of schedules to choose from. However, as the number of offered assignment schedules increases, the more difficult it becomes to identify the optimal schedule, and thus ensure truthfulness, as the problem begins to represent a full blown scheduling problem. Hence, there is an implicit trade-off between the distortion of any algorithm and its guarantees with respect to truthfulness.
\begin{table}
        \centering
        \begin{tabular}{|c|cccc|}
            \hline
             & 1 & 2 & 3 & 4 \\
             \hline
            $1:\: a \succ b \succ c$ & a & c & a & c \\
            $2:\: b \succ c \succ a$ & b &   & b &  \\
            $3:\: b \succ c \succ a$ & c &   & c &  \\
            \hline
        \end{tabular}
        \quad
        \begin{tabular}{|c|cccc|}
            \hline
               & 1 & 2 & 3 & 4  \\
             \hline
             a &  1 & \ding{55} & 1 & \ding{55} \\
             b &  2 & \ding{55} & 2 & \ding{55} \\
             c &  3 & 1 & 3 & 1 \\
             \hline
        \end{tabular}
        \caption{The  assignment of services under permutation $\sigma$ for the first four time steps when each agent reports their true preferences. The left table describes the sequence of services assigned to each agent, whilst the right table describes the sequence of agents assigned to each service. We use \ding{55} to denote that a given service is blocked on the corresponding time step.}
        \label{tab:true}
    \end{table}

    \begin{table}
        \centering
        \begin{tabular}{|c|cccc|}
            \hline
             & 1 & 2 & 3 & 4 \\
             \hline
            $1:\: a \succ b \succ c$ & b & a & b & a \\
            $2:\: b \succ c \succ a$ & c & c  & c & c   \\
            $3:\: b \succ c \succ a$ & a &   &  &   \\
            \hline
        \end{tabular}
        \quad
        \begin{tabular}{|c|cccc|}
            \hline
               & 1 & 2 & 3 & 4  \\
             \hline
             a &  3 & 1 & \ding{55} & 1  \\
             b &  1 & \ding{55} & 1 & \ding{55}\\
             c &  2 & 2 & 2 & 2 \\
             \hline
        \end{tabular}
        \caption{The  assignment of services under permutation $\sigma$ for the first four time steps when agents $2$ and $3$ report their preferences truthfully, and agent $1$ misreports $b \succ a \succ c$. The left table describes the sequence of services assigned to each agent, whilst the right table describes the sequence of agents assigned to each service. We use \ding{55} to denote that a given service is blocked on the corresponding time step.}
        \label{tab:lie}
    \end{table}

\section{Additional Pseudocode}
In this section, for clarity, we provide additional pseudocode regarding the subroutines used by \texttt{RRSD}, \texttt{DRRSD}, and \texttt{BRRSD}. Pseudocode for available(), which checks whether a service can be assigned to an agent, is given in Algorithm \ref{alg:avai}. The pseudocode for earliest(), which checks the earliest time step in which an agent can be assigned a service, is given in Algorithm \ref{alg:earliest}. The pseudocode for overrun(), which checks whether an assignment of a service in one block of \texttt{DRRSD} will cause a blocking delay in the next block, is given in Algorithm \ref{alg:overrun}. Lastly, the pseudocode for buildJobList(), which specifies the agent-service assignments for \texttt{BRRSD} to complete during the exploration phase, is given in Algorithm \ref{alg:job}.
\begin{algorithm}[t!] 
\SetAlgoLined
\SetKwFunction{Av}{available}
\SetKwProg{Fn}{Function}{:}{}
\Fn{\Av{$M$, $i_{1}$, $j$, $t_{0}$}}{
    \If{$M(t_{0}, i_{1}) \neq 0$}{
        \Return{false}
    }
    \For{$i = 1, \dots, n$}{
        \For{$t = t_{0} - D_{i, j}+ 1, \dots, t_{0} + D_{i, j} - 1$}{
            \If{$M(t, i) = j$}{
                \Return{false}
            }
        }
    }
    \Return{true}
}
\textbf{End Function} \\
\Fn{\Av{$M$, $i$, $j$, \text{start}, \text{end}}}{
    \For{$t=\text{start}, \dots, \text{end}$}{
        \If{$\text{available}(M, i, j, t)$}{
            \Return{\text{true}}
        }
    }
    \Return{\text{false}}
}
\textbf{End Function} \\
\Fn{\Av{$M$, $i$, $j$}}{
    $\text{start} = 1$ \\
    $\text{end} = T$ \\
    \Return{\text{available}($M$, $i$, $j$, \text{start}, \text{end})}
}
\textbf{End Function}
 \caption{Different versions of the auxiliary process available\label{alg:avai}}
\end{algorithm}

\begin{algorithm}[t!] 
\SetAlgoLined
\SetKwFunction{ear}{earliest}
\SetKwProg{Fn}{Function}{:}{}
\Fn{\ear{$M$, $i$, $j$, \text{start}, \text{end}}}{
    \For{$t = \text{start}, \dots, \text{end}$}{
        \If{$\text{available}(M, i, j, t)$}{
            \Return{$t$}
        }
    }
    \tcp{Failure state}
    \Return{0}
}
\textbf{End Function} \\
\Fn{\ear{$M$, $i$, $j$}}{
    $\text{start} = 1$ \\
    $\text{end} = T$ \\
    \Return{$\text{earliest}(M, i, j , start, end)$}
}
\textbf{End Function}

 \caption{Different versions of the auxiliary process earliest\label{alg:earliest}}
\end{algorithm}

\begin{algorithm}[t!] 
\SetAlgoLined
\SetKwFunction{ear}{overrun}
\SetKwProg{Fn}{Function}{:}{}
\Fn{\ear{$M$, $i$, $j$, $t$, \text{end}}}{
    \If{$t + D_{i, j} - 1 \geq \text{end}$}{
        \Return{false}
    }
    \Return{true}
}
\textbf{End Function}

\caption{The auxiliary process overrun\label{alg:overrun}}
\end{algorithm}

\begin{algorithm}
\SetAlgoLined
\SetKwFunction{job}{buildJobList}
\SetKwProg{Fn}{Function}{:}{}
\Fn{\job{$n$, $s$, \text{repeat}}}{
    list = [] \\
    \For{$i=1, \dots, n$}{
        \For{$j = 1, \dots, s$}{
            \For{$k=1, \dots, \text{repeat}$}{
                list.append(($i$, $j$))
            }
        }    
    }
}
\textbf{End Function}

 \caption{The auxiliary process buildJobList\label{alg:job}}
\end{algorithm}


\end{document}